%% file: DolevDuboisPotopTixeuil.tex
\documentclass[11pt]{article}

\usepackage{latexsym,amsfonts,amssymb,float}

\usepackage{epic,eepic,epsf}
\usepackage{subfigure}
\usepackage{wrapfig}
\usepackage{verbatim}
\usepackage{algorithmic}
\usepackage{algorithm}

\usepackage{stmaryrd}

\usepackage{fullpage}

\newtheorem{theorem}{Theorem}
\newtheorem{lemma}{Lemma}

\newenvironment{proof}{\noindent\textbf{Proof:}}{\hfill$\Box$}

\floatstyle{ruled}
\newfloat{algorithm}{Algorithm}{tbhp}[section]

\begin{document}

\renewcommand{\thefootnote}{\fnsymbol{footnote}}

\title{Stabilizing Data-Link over non-FIFO Channels\\ with Optimal Fault-Resilience}

\author{
Shlomi Dolev\footnotemark[1] \and 
Swan Dubois\footnotemark[2] \and 
Maria Potop-Butucaru\footnotemark[2] \and 
S\'ebastien Tixeuil\footnotemark[3]}
\date{}

\maketitle

\footnotetext[1]{Department of Computer Science, Ben-Gurion University of the Negev, Beer-Sheva, 84105, Israel. Email: {\tt dolev@cs.bgu.ac.il}. The work started while this author was a visiting professor at LIP6.
Research supported in part by the ICT Programme of the European Union under contract
number FP7-215270 (FRONTS), Deutsche Telekom, US Air-Force
and Rita Altura Trust Chair in Computer Sciences.}

\footnotetext[2]{UPMC Sorbonne Universit\'es \& INRIA, France.}

\footnotetext[3]{UPMC Sorbonne Universit\'es \& IUF, France. This work is supported in part by ANR projects SHAMAN, ALADDIN, and SPADES}

\begin{abstract}
Self-stabilizing systems have the ability to converge to a correct behavior when started in any configuration. Most of the work done so far in the self-stabilization area assumed either communication via shared memory or via FIFO channels. 

This paper is the first to lay the bases for the design of self-stabilizing message passing algorithms over unreliable non-FIFO channels. We propose an optimal stabilizing data-link layer that emulates a reliable FIFO communication channel over unreliable capacity bounded non-FIFO channels.   
\end{abstract}

\maketitle

\renewcommand{\thefootnote}{\arabic{footnote}}

\section{Introduction}

Self-stabilization~\cite{D74j,D00b,T09bc} is one of the most versatile techniques to sustain availability, reliability, and serviceability in modern distributed systems. After the occurrence of a catastrophic failure that placed the system components in some arbitrary global state, self-stabilization guarantees recovery to a correct behavior in finite time without external (\emph{i.e.} human) intervention. 

As self-stabilization is usually considered a hard property to satisfy, most related works used a simple communication model where processes can determine the current state of every neighbors (and update their own state accordingly) in an atomic manner (this model is referred to in the literature as the \emph{state model} or systems with \emph{central/distributed daemon}). Asynchronous message passing is a more realistic way, compared to the state model, for the communication of processes in distributed systems. In such settings processes communicate by exchanging messages, where sending and receiving message are two separate atomic actions. Transformers for shared memory protocols to act in message passing systems, assuming the existence of FIFO channels, have been suggested, see \emph{e.g.} \cite{DIM93j,D00b}. At the core of those transformers are the \emph{data-link} protocols, that permit to reliably exchange information between neighboring processes in the message passing model. In addition, several self-stabilizing protocols (\emph{i.e.} \cite{DT06c,AADDPT10c}) that are directly written in the message-passing model use an underlying data-link protocol as a building block.

\paragraph{Related Works.} The most studied data-link protocol, namely the alternating bit protocol (ABP), was proved to satisfy some stabilization properties~\cite{AB93j,DIM97j,BGM93j}: in any execution of ABP, there exists a suffix that satisfies the specification (\emph{i.e.} the ABP is \emph{pseudo-stabilizing}). However, the impossibility to bound the amount of time before this suffix is reached makes the ABP unsuitable for most tasks. In~\cite{GM91j,DIM93j}, Gouda and Multari and Dolev, Israeli, and Moran independently prove that for a wide class of problems (including data-link construction) guaranteeing self-stabilization when channels have unbounded initial capacity requires some kind of unboundedness in the protocol (either unbounded memory in~\cite{GM91j}, the existence of some aperiodic function~\cite{AB93j}, or access to a probabilistic variable~\cite{AB93j}). In other words, those approaches require to implement unbounded capacities with finite memory, and are thus unlikely to be actually used in real systems. Also, the expected time before reaching a stable global state depends on the initial contents of communication channels, and is thus unbounded.

Most recent works took the more realistic approach of assuming channels with bounded initial capacity. The token passing protocol in~\cite{DIM97j} can be used as a self-stabilizing ABP on bounded channels and only uses bounded memory. Howell \emph{et al.} \cite{HNM99c} provide another data-link protocol over bounded channels with the additional desirable property that the underlying communication channels are unreliable (\emph{i.e.} they may loose or duplicate messages). Later, Varghese~\cite{V00j} presented self-stabilizing solutions for a wide class of problems (including data-link) in the same setting using only bounded memory. The FIFO ordering is crucial for the stabilization since solution relies on the fact that a sequence number that is unique in the system is eventually generated and flushes every stale message in transit. A common drawback of all aforementioned self-stabilizing data-link solutions is that they assume a FIFO order on messages in the underlying communication channels. 

A notable exception are the protocols provided in \cite{BK97j} that assumed a non-FIFO message passing system. The main difference with our approach stands in the fact that their system is enhanced with some failure detector whereas we assume a fully asynchronous system.

Another drawback of previously mentioned self-stabilizing data-link solutions is that they do not consider the quantitative impact of faults from the perspective of the upper layer protocol (\emph{i.e.} the layer that actually uses the data-link). Indeed, starting from an arbitrary global state where channels may initially contain messages 
of arbitrary content, being able to bound the number of messages sent that are lost or duplicated, or the number of fake messages that are actually delivered to the destination is a very important matter. The bound on the number of faulty messages delivered by a data-link protocol is an important criteria for the data-link usability in larger application, in order to ensure the fault-resiliency of the global protocol stack. To our knowledge, only~\cite{DT06c,DDNT10j} addresses, to some extent, this concern. A snap-stabilizing data-link (and global reset) for bounded capacity FIFO channels appears in \cite{DT06c}. In \cite{DDNT10j} a snap-stabilizing solution to the propagation of information with feedback (PIF) problem is presented. The solution can be seen as a data-link protocol when reduced to a 2-processes system. Snap-stabilization implies that any message that is actually sent by the sender process is eventually received by the receiver process, so the number of lost messages is $0$. However, we cannot provide bounds on the number of duplications of a given message or on the number of ghost messages (that is, messages that are not sent by the sender  but received by the receiver due to the arbitrary content of communication channels in the initial configuration). Concerning the self-stabilizing protocols, only an order of magnitude on those numbers can be inferred from the stabilization time (if $m$ messages sent or received are required to enter a legitimate global state from any arbitrary initialization, then at most $m$ messages could be lost, duplicated, or wrongly delivered). To our knowledge, the question of fault-resilience optimality for data-link protocols has never been raised before, although it has important practical consequences. 

\paragraph{Our contribution.} Our contribution in this paper is twofold:
\begin{enumerate}
\item We define complexity metrics that are related to the fault-resilience of data-link protocols, and present impossibility results in the context of self-stabilization (\emph{i.e.} the ability to recover from any arbitrary initial global state). In particular, we prove that no data-link protocol can prevent one message duplication, the delivery of a single fake message, or the reordering of a single message.
\item We present a data-link protocol that is optimal with respect to all presented fault-resilience metrics. Moreover, unlike previous self-stabilizing solutions that operate assuming the underlying communication channels preserve FIFO ordering, the channels we consider may indeed reorder messages, having some of them remain in the channel for an arbitrary long time. The strong fault-resilience property exhibited by our protocol makes it particularly suitable for inclusion as a building block in more complex applications.
\end{enumerate}

\paragraph{Paper organization.} The paper is organized as follows. Section \ref{sec:model} proposes the network model and hypothesis and then, the data-link problem specification. Section \ref{sec:impossibility} introduces three lower bounds results that justify our optimality claim. In Section \ref{sec:solution}, we propose our optimal stabilizing data-link protocol altogether with its correctness proof.

\section{Model}\label{sec:model}

\subsection{System Model}

A \emph{message-passing distributed system} consists of $n$ {\em processes}, $p_0,p_1,p_2,\ldots,p_{n-1}$, connected by {\em communication links} through which messages are sent and received. Two processes connected through a communication link are referred in the following as neighboring processes.

As emphasized in \cite{AB93j} the purpose of a data-link protocol is to reliably transmit messages from one end of a communication medium (link) to the other end. Ideally, messages have to arrive without duplication or loss and in the order they have been sent. Therefore, we focus in the following on the communication between two neighboring processes $p_i$ and $p_j$ where $p_i$ acts as the sender and $p_j$ acts as the receiver. The communication link between the two processes $p_i$ and $p_j$ is denoted in the following $(p_i,p_j)$ and is composed of two virtual directed channels $(i,j)$ and $(j,i)$. The channel $(i,j)$ is used to send messages from $p_i$ to $p_j$ while the channel $(j,i)$ is used to send acknowledgments from $p_j$ to $p_i$. In systems where $p_j$ is also message sender, two additional virtual channels are used to carry the messages from $p_j$ to $p_i$ and acknowledgments from $p_i$ to $p_j$.

We assume in the following that the capacity of each directed channel is $c$ packets (\emph{i.e.} low level messages). Note that in the scope of self-stabilization, where the system copes with an arbitrary starting configuration, there is no deterministic data-link simulation that uses bounded memory when the capacity of channels is unbounded~\cite{GM91j,DIM97j}.

The channels are non-FIFO and not necessarily reliable (\emph{i.e. } packets may not follow the FIFO order and may be lost). Additionally, their delivery time is unbounded. That is, any non lost packet is received in a finite but unbounded time. Each channel $(i,j)$ is \emph{weakly fair} in the sense that if the sender sends infinitely often a packet on the channel, then the receiver receives this packet an infinite number of time. Sending a packet to a channel whose capacity is exhausted (\emph{i.e.} the channel already contains $c$ packets) results in loosing a packet (either a packet already in the channel or the packet being sent). 

As we deal with arbitrary initial corruption, a channel may initially contain up to $c$ ghost packets (\emph{i.e.} packets that have never been sent and contain arbitrary content).

A processor is modeled by a state machine that executes \emph{steps}. Channels are modeled as sets (rather than queues to reflect the non-FIFO order). For example, the $c$-bounded channel $(i,j)$ (used to send messages from $p_i$ to $p_j$) is modeled by a $c$-sized set denoted by $s_{ij}$.

In each step, a processor changes its local state (\emph{i.e.} the state of its local memory), and executes a single communication operation, which is either a {\em send}  operation or a {\em receive} operation. The communication operation changes the state of an attached channel. In case the communication operation is a send operation from $p_i$ to $p_j$ then $s_{ij}$ is a union of $s_{ij}$ in the previous state with the sent packet. If the obtained union does not respect the bound $|s_{ij}|\leq c$ then an arbitrary message in the obtained union is deleted. In case the communication operation is a receive operation of a (non null) packet $m$ ($m$ must exist in $s_{ji}$ of the previous state), then $m$ is removed from $s_{ji}$. A receive operation by $p_i$ from $p_j$ may result in a null packet even when the $s_{ji}$ is not empty, thus allowing unbounded delay for any particular packet. Packet losses are modeled by allowing spontaneous packet removals from the set.
 
A \emph{configuration} of the system is the product of the local states of processes in the system and of their incident channels.

An {\em execution} is a sequence of configurations, $E=(C_1,C_2,\ldots)$ such that $C_i$, $i>1$, is obtained from $C_{i-1}$ when at least one process in the system executes a step.

\subsection{Problem Specification}\label{sub:specification}

The specification we provide in this section is borrowed from \cite{L96b} but we adapt it to the self-stabilizing context. In particular, we introduce the idea to bound the number of lost, duplicated, ghost and re-ordered messages by some constants.

Consider a system of two processors $p_i$ and $p_j$. A distributed application needs to send some messages from $p_i$ to $p_j$. We say that the application layer of $p_i$ \emph{sends} a message when it requests the communication protocol to carry this message to $p_j$. This message is \emph{delivered} to $p_j$ when the communication protocol releases this message to the application layer of $p_i$. A \emph{ghost} message is a message delivered to $p_j$ whereas $p_i$ did not send it previously (due to the arbitrary content of communication channels in the initial configuration). A \emph{duplicated} message is a message that is delivered several times to $p_j$ whereas $p_i$ sent it only once. A message is \emph{lost} when $p_i$ sends it but $p_j$ never delivers it. A message $m$ is \emph{reordered} when it is delivered to $p_j$ before a message $m'$ whereas $m$ has been sent after $m'$ by $p_i$. Intuitively, the goal of a Stabilizing Data-Link protocol is to provide a communication black box that ensures some properties on the number of lost, duplicated, ghost and reordered messages starting from any arbitrary configuration. In the sequel, we formally specify the Stabilizing Data-Link problem

We associate to any execution $E$ the sequence $S(E)=m_0m_1m_2\ldots$ of messages sent by $p_i$ in $E$ and the sequence $R(E)=m'_0m'_1m'_2\ldots$ of messages delivered to $p_j$ in $E$. Note that we consider that all sent messages are different (even if their actual content are identical, we can distinguish them as external observer of the system). We introduce the following notations. For any sequence $W$ and any integers $i$ and $j$, $W^j$ is the prefix of $W$ of length $j$ and $W_i$ is the suffix of $W$ such that $W=W^{i-1}W_i$. The notation $\epsilon$ denotes the empty sequence. For example, $R(E)^0=\epsilon$. For any message $m$, we define the $m^*$ as the repetition of $m$ an arbitrary number of times (possibly 0). For any sequence $W$, the sequence $W^*$ is the result of the application of the $^*$ operator to each message of $W$. 

For any non negative integers $\alpha$, $\beta$, $\gamma$, and $\delta$, the \textbf{($\alpha,\beta,\gamma,\delta$)-Stabilizing Data-Link communication} over $c$-bounded channels satisfies the following properties starting from an arbitrary configuration (with $p_i$ and $p_j$ being respectively the sender and the receiver) for any execution $E$:
\begin{itemize}
\item \textbf{$\alpha$-Loss:} The first $\alpha$ messages sent by $p_i$ (in the worst case) may be lost.
\[\exists a\leq\alpha, \forall m\in S(E)_a, m\in R(E)\]
\item \textbf{$\beta$-Duplication:} The first $\beta$ messages delivered to $p_j$ (in the worst case) may be duplicated ones.
\[\exists b\leq\beta, \forall m\in S(E), \big|\{m'_i=m|m'_i\in R(E)\}\big|> 1\Rightarrow m\in R(E)^b\]
\item \textbf{$\gamma$-Creation:} The first $\gamma$ messages delivered to $p_j$ (in the worst case) may be ghost messages.
\[\exists c\leq\gamma, \forall m\in R(E), m\notin S(E)\Rightarrow m\in R(E)^c\]
\item \textbf{$\delta$-Reordering:} The first $\delta$ messages delivered to $p_j$ (in the worst case) may be reordered.
\[\exists d\leq \delta, R(E)_d=S(E)^*\]
\end{itemize}

In the following section, we show that it is impossible to perform a ($\alpha,\beta,\gamma,\delta$)-Stabilizing Data-Link communication with $\beta=0$, $\gamma=0$, or $\delta=0$. Then, we can deduce that a  ($0,1,1,1$)-Stabilizing Data-Link communication achieves optimal fault-resiliency. The above definitions imply that such a communication protocol ensures that $R(E)=S(E)$ or $R(E)=m.S(E)$ (where $m$ is an arbitrary message, it may be present in $S(E)$) for any execution $E$. In other words, the sequence of received messages by $p_j$ is identical to the sequence of emitted messages by $p_i$ excepted the first delivery in the worst case.

\section{Lower Bounds}\label{sec:impossibility}

In this section, we propose three impossibility results related to the possible values for the parameters $\beta$, $\gamma$, and $\delta$. We prove that the lower bounds for $\beta$, $\gamma$, and $\delta$ parameters is $1$. These results confirm the claim that the protocol we propose is optimal since it implements a $(0,1,1,1)$-Stabilizing data-link.

\begin{theorem}\label{th:impGamma}
There exists no $(\alpha,\beta,\gamma,\delta)$-Stabilizing Data-Link communication algorithm over $c$-bounded channels with $\gamma=0$.
\end{theorem}

\begin{proof}
By contradiction, let $\mathcal{A}$ be any $(\alpha,\beta,0,\delta)$-Stabilizing Data-Link communication algorithm over $c$-bounded channels must have an instruction that delivers messages to the receiver processor. As the program counter may be corrupted and channels may contain up to $c$ ghost messages in the initial configuration, the receiver processor may execute this instruction during the first step of an execution $E$. In consequence, the first message of $R(E)$ may be a ghost message $m$. Hence, we can assume that $R(E)^1=m$.

It is possible to construct the execution $E$ such that $m\notin S(E)$. In conclusion, we have: $\exists m\in R(E), m\notin S(E)\wedge m\notin R(E)^0=\epsilon$ (recall that $\epsilon$ denotes the empty sequence). This is contradictory with the 0-Creation property of $\mathcal{A}$ and implies that $\gamma\geq 1$ for any $(\alpha,\beta,\gamma,\delta)$-Stabilizing Data-Link communication algorithm over $c$-bounded channels.
\end{proof}

\begin{theorem}\label{th:impBeta}
There exists no $(\alpha,\beta,\gamma,\delta)$-Stabilizing Data-Link communication algorithm over $c$-bounded channels with $\beta=0$.
\end{theorem}

\begin{proof}
By contradiction, let $\mathcal{A}$ be any $(\alpha,0,\gamma,\delta)$-Stabilizing Data-Link communication algorithm over $c$-bounded channels. Following Theorem \ref{th:impGamma}, we have $\gamma>0$. This implies that the first message delivered to $p_j$ in an execution $E$ by $\mathcal{A}$ may be a ghost message $m$. Hence, we can assume that $R(E)^1=m$.

It is possible to construct the execution $E$ such that the first (real) message sent by $p_i$ to $p_j$ and delivered to $p_j$ by $\mathcal{A}$ is the same message $m$. This message has been sent by $p_i$ only once but has been delivered to $p_j$ at least twice. In conclusion, we have: $\exists m\in S(E),\big|\{m'_i=m|m'_i\in R(E)\}\big|>1 \wedge m\notin R(E)^0=\epsilon$ (recall that $\epsilon$ denotes the empty sequence). This is contradictory with the 0-Duplication property of $\mathcal{A}$ and implies that $\beta\geq 1$ for any $(\alpha,\beta,\gamma,\delta)$-Stabilizing Data-Link communication algorithm over $c$-bounded channels.
\end{proof}

\begin{theorem}\label{th:impDelta}
There exists no $(\alpha,\beta,\gamma,\delta)$-Stabilizing Data-Link communication algorithm over $c$-bounded channels with $\delta=0$.
\end{theorem}

\begin{proof}
By contradiction, let $\mathcal{A}$ be any $(\alpha,\beta,\gamma,0)$-Stabilizing Data-Link communication algorithm over $c$-bounded channels. Following Theorem \ref{th:impGamma}, we have $\gamma>0$.  This implies that the first message delivered to $p_j$ by $\mathcal{A}$ in an execution $E$ may be a ghost message $m$. Hence, we can assume that $R(E)^1=m$.

It is possible to construct the execution $E$ such that $S(E)^{\alpha+2}=m_0m_1\ldots m_{\alpha-1}m_\alpha m$ and $\forall i\in\{0,\ldots,\alpha\},m_i\neq m$. As $\mathcal{A}$ satisfies the $\alpha$-Loss and the 0-Reordering properties, it follows that $\exists i\in\{0,\ldots,\alpha\},R(E)^1=m_i$ (otherwise, we have a contradiction since either $\mathcal{A}$ lost at least $\alpha+1$ messages or reordered at least one message, that is contradictory). As $m_i\neq m$, we obtain a contradiction that shows that $\delta\geq 1$ for any $(\alpha,\beta,\gamma,\delta)$-Stabilizing Data-Link communication algorithm over $c$-bounded channels.
\end{proof}

In the next section, we present a protocol that is optimal with respect to $\alpha$, $\beta$, $\gamma$, and $\delta$ parameters. That is, our protocol satisfies the $(0,1,1,1)$-Stabilizing Data-Link specification.

\section{A $(0,1,1,1)$-Stabilizing Data-Link Protocol}\label{sec:solution}

\subsection{Presentation of the Protocol}

\paragraph{Key ideas of the protocol.} The rationale of the protocol consists in adding safety extensions to the well-known alternating bit protocol (\emph{a.k.a.} ABP). The concept used in the design of the data-link protocol is to let the sender use a mechanism based on the capacity $c$ of communication channels so that the sender can ensure the execution of an operation in the receiver side. More precisely, the receiver acts only upon receiving a packet from the sender. The sender may repeatedly send a particular packet, and each time the receiver receives a packet it acknowledges the packet arrival.

First, the receiver can deliver a message only if $c+1$ copies of this message have been previously received: this ensures that at least one of them is genuine (\emph{i.e.} was actually sent by the sender). Moreover, a message is delivered only if the expected bit alternates with the one of the previously received message (similarly to the ABP) in order to ensure that no message is duplicated. Indeed, the sender may still send copies of the message with the same alternating bit value until it receives a sufficient number of acknowledgments. 

Second, the sender will expect for each message sent at least $3c+2$ acknowledgments with a matching alternating bit. As up to $c$ acknowledgments could be ghost, this implies that $2c+2$ of these acknowledgments were actually sent by the receiver. One such acknowledgment could be sent by the received due to bad initialization, $c$ of them could be due to $c$ initial ghost messages in the reverse direction, and the remaining $c+1$ can only originate from genuine messages from the sender, that triggered a delivery at the receiver. 

At this stage, the protocol does not ensure the $0$-Loss property due to the use of the alternating bit. Indeed, if the alternating bit values of the sender and of the receiver are not synchronized at the first delivery, the receiver drops the first message. To avoid this message loss, the sender alternates between actual messages and synchronization messages. In other words, to send a message $m$, the sender first sends a synchronization message (denoted by $<SYNCHRO>$) until it receives $3c+2$ acknowledgments of this synchronization message and then send the actual message $m$ until it receives $3c+2$ acknowledgments of $m$. It follows that only the synchronization message may be lost and the actual message is always delivered to the receiver.

\begin{figure*}[t]
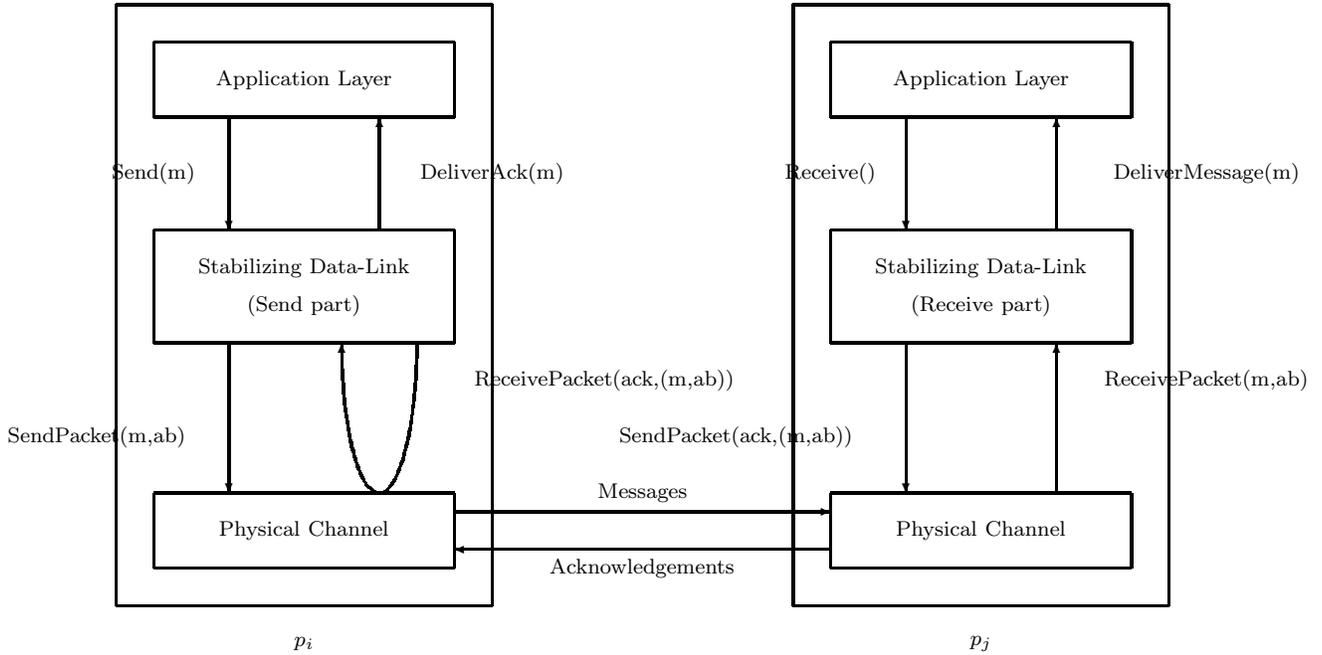

\noindent \begin{centering} \scriptsize{\include{schema}}
  \par\end{centering}
 \caption{General organization of our system.}
\label{fig:schema}
\end{figure*}

\paragraph{General organization of the system.}
Our system is organized as follows. The application layer generates messages to be send from $p_i$ to $p_j$. To perform this goal, it invokes our stabilizing data-link protocol. Furthermore, this layer invokes procedures provided by the physical channel.

In more details, the stabilizing data-link protocol is composed of two functions: \textbf{Send} (which is executed on the sender side) and \textbf{Receive} (which is executed on the receiver side). When the application layer on the sender side wants to send a message $m$, it invokes \textbf{Send(m)}. \textbf{Send} procedure is blocking, that is if \textbf{Send} is already in execution, the application layer waits its termination whereas the \textbf{Receive} function is always executed on the receiver side. When the \textbf{Receive} function has a message to deliver at the application layer on the receiver side, it executes \textbf{DeliverMessage($m$)} that transmits $m$ to the application layer. When the \textbf{Receive} function wants to discard a synchronization message (since this kind of messages is useless to the application layer), it uses the \textbf{DropMessage} function that only deletes the message. Finally, each delivered message is acknowledged to the application layer on the sender side by \textbf{DeliverAck($m$)}.

Functions \textbf{Send} and \textbf{Receive} must interact with the physical channel in order to exchange messages. For this, we assume that the channel provides two operations. First, it provides an operation to send a message or an acknowledgment, respectively \textbf{SendPacket($m$,$ab$)} and \textbf{SendPacket(ack,($m$,$ab$))} where $m$ is the message and $ab$ its alternating bit value. This operation puts $m$ (or its acknowledgment) in the channel if it is possible (if this operation leads to more than $c$ messages in the channel, one of them is arbitrarily deleted). Second, it provides an operation to receive a message or an acknowledgment, respectively \textbf{ReceivePacket($m$,$ab$)} and \textbf{ReceivePacket(ack,($m$,$ab$))} where $m$ is the message and $ab$ its alternating bit value. On the receiver side, \textbf{ReceivePacket($m$,$ab$)} is executed when the channel has a message to deliver and when \textbf{Receive} is not in execution. It sets then $m$ and $ab$ to actual values of the delivered message. In other words, the reception (for the data-link protocol) on the receiver side is message-driven. On the sender side, \textbf{ReceivePacket(ack,($m$,$ab$))} is executed by the data-link protocol and does polling. That is it checks whether the first waiting message in the channel (if any) matches with an acknowledgment of the parameter ($m$,$ab$). It returns \textbf{true} if this is the case, \textbf{false} otherwise. In any case, the first waiting message (if any) is deleted from the channel. The architecture of our system is summarized in Figure \ref{fig:schema}.

\paragraph{Detailed presentation of the protocol.}
Our $(0,1,1,1)$-stabilizing data-link protocol $\mathcal{SDL}$ is presented as Figure \ref{algo:SDL}. In the following, we provide details about the two functions \textbf{Send} and \textbf{Receive}.

The function \textbf{Send} takes a message $m$ as parameter and stores the current alternating bit value in the variable $ab$. First, it alternates the value of $ab$ (line 01) before sending a synchronization message (line 02) using an auxiliary function \textbf{SendMessage}. Then, lines 03 and 04 repeat these instructions with the message $m$. Once the last invocation of \textbf{SendMessage} returns, it delivers to the application layer the acknowledgment of $m$ using \textbf{DeliverAck}. Now, let us describe the auxiliary function \textbf{SendMessage}. This function repeatedly (while loop of line 02) sends its parameter message $m$ (line 03) until receiving $3c+2$ acknowledgment for this message (line 04-05).

The function \textbf{Receive} takes no parameter and uses two variables. The first one is the alternating bit value of the last delivered or dropped message stored in $last\_delivered$ and the second one is a queue $Q$ that stores the number of receptions of at most $c+1$ different messages. Each element of this queue is a  $3$-tuple $(m,ab,count)$, where $m$ is a message, $ab$ is an alternating bit value, and $count$ is an integer denoting the number of packets $(m,ab)$ received for the corresponding $m$ and $ab$ since the last \textbf{DeliverMessage} or \textbf{DropMessage} occurred. The queue $[]$ operator takes a message $m$ and a boolean $b$ as operands, and either enqueues $(m,ab,0)$ (if $(m,ab,*)$ is not present in $Q$, then if the queue contained $c+1$ elements, the last element of the queue is dequeued) or returns a pointer to the $count$ value associated to $m$ and $ab$ in $Q$. Any time a tuple value is changed in the queue, this tuple is promoted at the top of the queue (in order to keep in memory the $c+1$ latest received messages), and the size of the queue does not change. The $\bot$ assignment to a queue $Q$ denotes the fact that $Q$ is emptied. At each reception of a message $(m,ab)$ (line 01), the corresponding entry in the queue is updated (or created if needed) by line 02. If $p_j$ already received $c+1$ copies of $m$ since the last \textbf{DeliverMessage} or \textbf{DropMessage} occurred (test on line 03) then the queue is emptied (line 10). Moreover, if the alternating bit value of the message is different from $last\_delivered$ (test on line 04), then the message is either delivered with \textbf{DeliverMessage} (line 06) or dropped with \textbf{DropMessage} (line 08) depending if it is a synchronization message or not (test on line 05). Then, the $last\_delivered$ value is updated by line 09. Finally, in any case, the message is acknowledged to the sender with line 11.

\begin{figure*}[htb!]
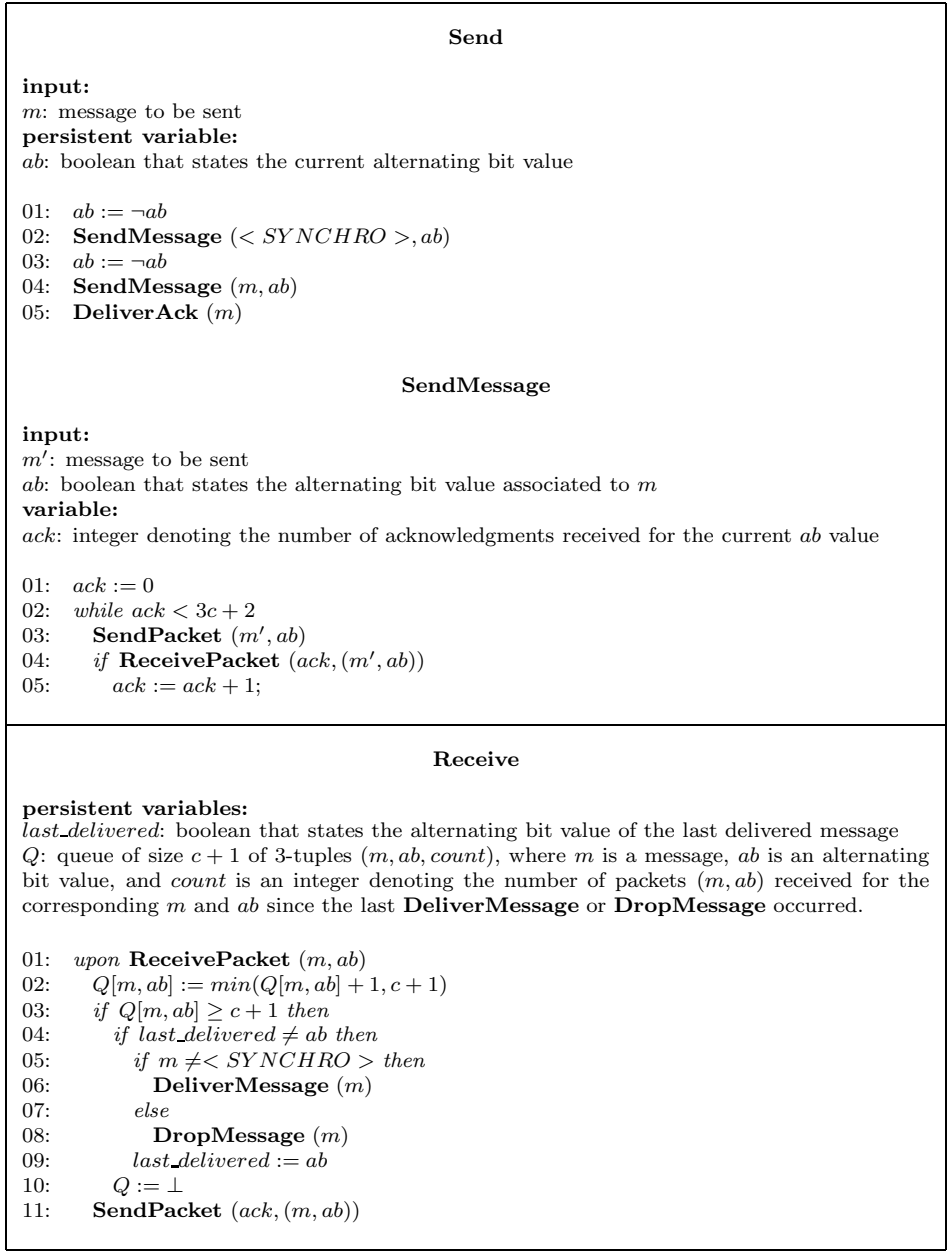

\scriptsize
\centering
\begin{tabular}{|p{4.75in}|}
\hline
\vspace*{0.1cm}
\begin{minipage}[t]{4.75in}
\centering
{\it\bf Send}
\begin{tabbing}
X: \= d \= d \= d \= d \= d \= d \= d \= d \= \kill

\textbf{input:}\\ 
$m$: message to be sent\\

\textbf{persistent variable:}\\
$ab$: boolean that states the current alternating bit value\\
\\

01: \>\>$ab := \neg ab$ \\
02: \>\>\textbf{SendMessage} $(<SYNCHRO>,ab)$ \\
03: \>\>$ab := \neg ab$ \\
04: \>\>\textbf{SendMessage} $(m,ab)$ \\
05: \>\>\textbf{DeliverAck} ($m$) \\
\end{tabbing}

\centering
{\it\bf SendMessage}
\begin{tabbing}
X: \= d \= d \= d \= d \= d \= d \= d \= d \= \kill

\textbf{input:}\\ 
$m'$: message to be sent\\
$ab$: boolean that states the alternating bit value associated to $m$\\

\textbf{variable:}\\
$ack$: integer denoting the number of acknowledgments received for the current $ab$ value\\
\\

01: \>\>$ack := 0$\\
02: \>\>\emph{while} $ack < 3c+2$ \\
03: \>\> \>\textbf{SendPacket} $(m',ab)$ \\
04: \>\> \>\emph{if} \textbf{ReceivePacket} $(ack,(m',ab))$ \\
05: \>\> \> \>$ack := ack+1$; \\

\end{tabbing}
\end{minipage}
\tabularnewline
 \hline
\vspace*{0.1cm}
\begin{minipage}[t]{4.75in}
\centering
{\it\bf Receive}
\begin{tabbing}
X: \= d \= d \= d \= d \= d \= d \= d \= d \= \kill 

\textbf{persistent variables:}\\
\begin{minipage}{4.75in}
$last\_delivered$: boolean that states the alternating bit
value of the last delivered message\\
$Q$: queue of size $c+1$ of $3$-tuples $(m,ab,count)$,
where $m$ is a message, $ab$ is an alternating bit value,
and $count$ is an integer denoting the number of packets
$(m,ab)$ received for the corresponding $m$ and $ab$
since the last \textbf{DeliverMessage} or \textbf{DropMessage} occurred.
\vspace*{0.1cm}
\end{minipage}
\\
\\
01: \>\>\emph{upon} \textbf{ReceivePacket} $(m,ab)$ \\
02: \>\> \>$Q[m,ab] := min(Q[m,ab]+1,c+1)$\\
03: \>\> \>\emph{if} $Q[m,ab] \geq c+1$ \emph{then} \\
04: \>\> \> \>\emph{if} $last\_delivered \neq ab$ \emph{then}\\ 
05: \>\> \> \> \>\emph{if} $m\neq<SYNCHRO>$ \emph{then}\\
06:\> \> \> \> \> \>\textbf{DeliverMessage} ($m$)\\
07: \>\> \> \> \>\emph{else}\\
08:\> \> \> \> \> \>\textbf{DropMessage} ($m$)\\
09:\> \> \> \> \>$last\_delivered := ab$\\
10: \>\> \> \>$Q := \bot$ \\
11: \>\> \>\textbf{SendPacket} $(ack,(m,ab))$\\
\end{tabbing}
\end{minipage}\\[1ex]
\hline
\end{tabular}
\normalsize
\caption{$\mathcal{SDL}$, a $(0,1,1,1)$-Stabilizing Data-Link protocol} 
\label{algo:SDL}
\end{figure*}

\subsection{Correctness Proof}

In this section, let $p_i$ and $p_j$ be two neighboring nodes that execute $\mathcal{SDL}$, $p_i$ being the sender and $p_j$ the receiver. Let $E=(C_1,C_2,\ldots)$ be an execution starting from an arbitrary configuration. 

We say that a message $m'$ is \emph{processed} by $p_j$ when $p_j$ executes \textbf{DeliverMessage} ($m'$) (line 06 of \textbf{Receive} function) if $m'$ is a normal message or when $p_j$ executes \textbf{DropMessage} ($m'$) (line 08 of \textbf{Receive} function) if $m'$ is a $<SYNCHRO>$ message.

First, we need two preliminaries results related to the result of the execution of the procedure \textbf{SendMessage} by $p_i$ depending on the configuration in which $p_i$ starts to execute this procedure.

\begin{lemma}\label{lem:firstabneqld}
When $p_i$ starts to execute \textbf{SendMessage} $(m',ab)$ in a configuration where $ab \neq last\_$ $delivered$, the message $m'$ (either a $<SYNCHRO>$ message or a normal message) and every message parameter to a subsequent invocation of \textbf{SendMessage} is processed by $p_j$ in a finite time.
\end{lemma}

\begin{proof}
Consider a configuration $C_k$ where $ab \neq last\_delivered$. Assume that $p_i$ starts to execute \textbf{SendMessage} $(m',ab)$ in $C_k$. By contradiction, assume $m'$ is never processed by $p_j$ in the remainder of $E$. That is, $p_j$ never executes lines 06 or 08 in the \textbf{Receive} procedure. In turn, tests on lines 03 or 04 never evaluate to true simultaneously. 

As $last\_delivered \neq ab$ in $C_k$ and $last\_delivered$ may change only when $m'$ is processed (line 09), we know that the test on line 04 is always true (since $m$ is never processed by assumption).

This implies that $Q[m',ab] \geq c+1$ never evaluates to true (test on line 03). This implies that the sender stops sending $(m',ab)$ before the $(m',ab)$ counter reached $c+1$, which is impossible. The reason is as follows. In order to stop sending the same message, $p_i$ must get $3c+2$ acknowledgments with the expected content $(ack,(m',ab))$. If such $3c+2$ acknowledgments are indeed received, this implies that the receiver issued at least $2c+2$ of those acknowledgments, and thus received $2c+2$ packets $(m',ab)$. Consider the first such packet $(m',ab)$ received by $p_j$. If there is no reset of $p_j$'s queue following this packet, the head of the queue now contains an entry $(m',ab,*)$ that can not be deleted until the receiver resets the entire queue. Indeed, at most $c$ packets are initially present in the receiver's input channel, that can create at most $c$ entries in the queue. Since the queue is of size $c+1$, the $(m',ab,*)$ tuple remains. Now, if the receiver sends $c+1$ packets $(ack,(m',ab))$, it implies that the receiver's queue for entry $(m',ab,*)$ was incremented $c+1$ times, which invalidates the assumption. It follows that $m'$ is processed in a finite time.

Note that after the processing of $m'$, $ab$ and $last\_delivered$ have the same value with the execution of the line 09 of \textbf{Receive} procedure. Hence the next invocation of the \textbf{SendMessage} primitive by $p_i$ will make the values $ab$ and $last\_delivered$ different. Applying the above reasoning, the lemma follows.
\end{proof}

\begin{lemma}\label{lem:firstab=ld}
When $p_i$ starts to execute \textbf{SendMessage} $(m',ab)$ in a configuration where $ab = last\_$ $delivered$, only $m'$ (either a $<SYNCHRO>$ message or a normal message) is not processed by $p_j$.
\end{lemma}

\begin{proof}
Consider a configuration $C_k$ where $ab=last\_delivered$. Assume that $p_i$ starts to execute \textbf{SendMessage} $(m',ab)$ in $C_k$.

Since the test in the line 04 of the \textbf{Receive} procedure evaluates to false, the processing of $m'$ is not executed. However, since $p_i$ keeps sending $m'$ and $p_j$ acknowledges these packets the \textbf{SendMessage} procedure returns. Note that $p_i$ executes line 01 or 03 of the \textbf{Send} procedure before the next invocation of \textbf{SendMessage} procedure.
 
It follows that the system reaches in a finite time a configuration where $ab \neq last\_delivered$. Then, Lemma \ref{lem:firstabneqld} implies that every message that is parameter of subsequent invocations of \textbf{SendMessage} is eventually processed by $p_j$. 
\end{proof}

Now, we can prove that $\mathcal{SDL}$ satisfies the four properties of the specification (see Section \ref{sub:specification}) starting from any configuration.

\begin{lemma}\label{lem:loss}
$\mathcal{SDL}$ satisfies the 0-Loss property.
\end{lemma}

\begin{proof}
Assume that $p_i$ has to send a message $m$ to $p_j$ starting from an arbitrary configuration. Note that proofs of Lemmas \ref{lem:firstabneqld} and \ref{lem:firstab=ld} imply that any invocation of the \textbf{Send} procedure eventually ends. This implies in turn that $p_i$ starts to execute \textbf{Send}$(m)$ in a finite time. 

Then, $p_i$ invokes first \textbf{SendMessage} with a $<SYNCHRO>$ message as parameter (see line 02 of the \textbf{Send} procedure). Note that this $<SYNCHRO>$ message may be lost if $ab=last\_delivered$ when $p_i$ starts to execute \textbf{SendMessage} by Lemma \ref{lem:firstab=ld}. 

Then, following Lemma \ref{lem:firstab=ld} that we have $ab\neq last\_delivered$ when $p_i$ starts to execute \textbf{SendMessage} with $m$ as parameter (see line 04 of the \textbf{Send} procedure) since it has executed line 03 of the \textbf{Send} procedure. By Lemma \ref{lem:firstabneqld}, it follows that $m$ is eventually processed by $p_j$. As $m$ is a normal message, this implies by definition that $m$ is delivered to $p_j$ in a finite time.

As this result holds whatever the state of the system when $p_i$ requests to send $m$, we obtain that $\forall m\in S(E), m\in R(E)$. It is sufficient to observe that $S(E)=S(E)_0$ to obtain the result.
\end{proof}

\begin{lemma}\label{lem:duplication}
$\mathcal{SDL}$ satisfies the 1-Duplication property.
\end{lemma}

\begin{proof}
By contradiction, assume that there exists an execution $E$ of $\mathcal{SDL}$ such that $\forall b\leq 1, \exists m\in S(E),\big|\{m'_i=m|m'_i\in R(E)\}\big|>1 \wedge m\notin R(E)^b$. In particular, this property is true for $b=1$. Hence, $\exists m\in S(E), \big|\{m'_i=m|m'_i\in R(E)\}\big|>1 \wedge m\notin R(E)^1$. In other words, there exists in $E$ a message $m$ sent by $p_i$ delivered several times to $p_j$. Moreover $m$ is not the first message received by $p_j$.

This implies that the line 06 in the \textbf{Receive} procedure is executed several times for the message $m$. It is impossible and the reason is the following. After the first delivery of $m$ the receiver empties the queue and makes $last\_delivered=ab$ (see proof of Lemma \ref{lem:firstab=ld}). Note that $p_i$ modifies $ab$ only when it stops to send $m$. Even if $p_i$ keeps invoking \textbf{SendPacket} $(m,ab)$ until it receives the $3c+2$ acknowledgments, none of these messages will be delivered since for each of them the test in line 04 in the \textbf{Receive} procedure returns false.

This contradiction implies that only the first message received by $p_j$ may be duplicate. The lemma follows. 
\end{proof}

\begin{lemma}\label{lem:creation}
$\mathcal{SDL}$ satisfies the 1-Creation property.
\end{lemma}

\begin{proof}
By contradiction, assume that there exists an execution $E$ of $\mathcal{SDL}$ such that $\forall c\leq 1, \exists m\in R(E), m\notin S(E) \wedge m\notin R(E)^c$. In particular, this property is true for $c=1$. Hence, $\exists m\in S(E), m\notin S(E)\wedge m\notin R(E)^1$. In other words, there exists in $E$ a message $m$ not sent by $p_i$ but delivered to $p_j$. Moreover $m$ is not the first message received by $p_j$.

Initially the channel $(i,j)$ may contain at most $c$ ghosts messages. In the worst case, the receiver's queue also contains an entry for each of the ghost with the counters initialized to $c$ or $c+1$.

Let $(g,ab)$ be the first ghost message received by $p_j$ with alternated bit set to $ab$. Let us study the two possible cases. First, assume that $ab \neq last\_delivered$. Then $p_j$ delivers $g$ (line 06 of \textbf{Receive} procedure) and empties the queue (line 10 of \textbf{Receive} procedure). Second, assume that $ab = last\_delivered$. Then $p_j$ does not deliver $g$ (due to the test of line 04 of \textbf{Receive} procedure) but it empties the queue (line 10 of \textbf{Receive} procedure).

In both cases, there is at most one ghost message delivered to $p_j$ and the queue has been emptied. In turn, it remains now at most $c-1$ ghosts messages in the channel $(i,j)$. If one of them is received by $p_j$ (after an invocation of \textbf{ReceivePacket}), its associated counter cannot reach the value $c+1$ (unless $p_i$ starts to send the same message but in this case, it is no longer a ghost message) since there are at most $c-1$ copies of the same message. Consequently, none of the $c-1$ remaining ghost messages can be delivered, that contradicts the construction of $m$ and proves the result.  
\end{proof}

\begin{lemma}\label{lem:reordering}
$\mathcal{SDL}$ satisfies the 1-Reordering property.
\end{lemma}

\begin{proof}
Following Lemma \ref{lem:creation}, $\mathcal{SDL}$ delivers at most one ghost message to $p_j$ in $E$. Let us consider the two following possible cases.

\begin{description}
\item[Case 1:] $\mathcal{SDL}$ delivers no ghost message to $p_j$ in $E$.\\
According to Lemmas \ref{lem:loss} and \ref{lem:duplication}, any message sent from $p_i$ is delivered to $p_j$ exactly once in this case. Now, observe that any message is delivered to $p_j$ between the beginning and the end of the corresponding execution of the procedure \textbf{Send} by $p_i$. Indeed, the message is delivered to $p_j$ when it receives the $(c+1)$-th copy of the message whereas $p_i$ waits to receive the $(3c+2)$-th acknowledgment of the message to stop sending it (see proof of Lemmas \ref{lem:firstabneqld} and \ref{lem:firstab=ld}). Since the \textbf{Send} procedure is blocking for $p_i$, $R(E)_0=R_E=S_E$ for any execution $E$ where $\mathcal{SDL}$ delivers no ghost message to $p_j$. Hence, $\exists d=0\leq 1,R(E)_d=S_E$.
\item[Case 2:] $\mathcal{SDL}$ delivers one ghost message to $p_j$ in $E$.\\
Assume that the ghost message delivered by $\mathcal{SDL}$ is $m$. Lemma \ref{lem:creation} allows us to state that $m$ is the first message delivered to $p_j$. Then, a similar reasoning to the one of case 1 allows us to conclude that $R(E)=m.S(E)$ for any execution $E$ where $\mathcal{SDL}$ delivers one ghost message $m$ to $p_j$ and then, $R(E)_1=S_E$. Hence, $\exists d=1\leq 1,R(E)_d=S_E$.
\end{description}

In both cases, we show that $\mathcal{SDL}$ satisfies the 1-Reordering property.
\end{proof}

Now, we can conclude on the following corollary of Lemmas \ref{lem:loss}, \ref{lem:duplication}, \ref{lem:creation} and \ref{lem:reordering}.

\begin{theorem}
$\mathcal{SDL}$ satisfies the $(0,1,1,1)$-Stabilizing Data-Link Communication specification.
\end{theorem}

\section{Conclusion}\label{sec:conclusion}

In this paper, we focused on stabilizing data-link protocols over channels of bounded capacity $c$. First, we introduced some measures for fault-resilience following the specification presented in \cite{L96b} that is suitable to the self-stabilizing setting. Then, we proved lowers bounds on these parameters. Finally, we proposed a stabilizing data-link protocol that emulates FIFO reliable links over unreliable bounded non-FIFO communication environment with an optimal fault-resilience. To achieve this optimal fault-resilience, our protocol sends $6c+4$ packets (and their corresponding acknowledgements) to deliver one message to the application layer.

Some interesting open questions follow. Is it possible to achieve optimal fault-resilience with a (significantly) lower message complexity for a given channel capacity $c$? Recently, some works on snap-stabilizing point-to-point communication \cite{CDV09ca,CDV09cb,CDLPV10c} across multiples hops have been presented in a coarse grained communication model. Is it possible to extend these results to the more realistic message passing model using our Stabilizing Data-link as a communication black box? If so, is it possible to provide optimal fault resilience as in the one hop case?

\bibliographystyle{plain}
\bibliography{biblio}

\end{document}

%% file: schema.tex
\ifx\JPicScale\undefined\def\JPicScale{1}\fi
\unitlength \JPicScale mm
\begin{picture}(147.5,85)(0,0)
\put(27.5,15){\makebox(0,0)[cc]{Physical Channel}}

\put(27.5,50){\makebox(0,0)[cc]{Stabilizing Data-Link}}

\put(27.5,75){\makebox(0,0)[cc]{Application Layer}}

\linethickness{0.3mm}
\put(7.5,80){\line(1,0){40}}
\put(7.5,70){\line(0,1){10}}
\put(47.5,70){\line(0,1){10}}
\put(7.5,70){\line(1,0){40}}
\linethickness{0.3mm}
\put(7.5,55){\line(1,0){40}}
\put(7.5,40){\line(0,1){15}}
\put(47.5,40){\line(0,1){15}}
\put(7.5,40){\line(1,0){40}}
\linethickness{0.3mm}
\put(7.5,20){\line(1,0){40}}
\put(7.5,10){\line(0,1){10}}
\put(47.5,10){\line(0,1){10}}
\put(7.5,10){\line(1,0){40}}
\linethickness{0.3mm}
\put(17.5,55){\line(0,1){15}}
\put(17.5,55){\vector(0,-1){0.12}}
\linethickness{0.3mm}
\put(17.5,20){\line(0,1){20}}
\put(17.5,20){\vector(0,-1){0.12}}
\linethickness{0.3mm}
\multiput(42.5,39.5)(0,0.5){1}{\line(0,1){0.5}}
\multiput(42.49,39)(0,0.5){1}{\line(0,1){0.5}}
\multiput(42.49,38.51)(0.01,0.5){1}{\line(0,1){0.5}}
\multiput(42.48,38.01)(0.01,0.5){1}{\line(0,1){0.5}}
\multiput(42.46,37.51)(0.01,0.5){1}{\line(0,1){0.5}}
\multiput(42.44,37.02)(0.02,0.49){1}{\line(0,1){0.49}}
\multiput(42.42,36.53)(0.02,0.49){1}{\line(0,1){0.49}}
\multiput(42.4,36.04)(0.02,0.49){1}{\line(0,1){0.49}}
\multiput(42.37,35.55)(0.03,0.49){1}{\line(0,1){0.49}}
\multiput(42.35,35.06)(0.03,0.48){1}{\line(0,1){0.48}}
\multiput(42.31,34.58)(0.03,0.48){1}{\line(0,1){0.48}}
\multiput(42.28,34.1)(0.04,0.48){1}{\line(0,1){0.48}}
\multiput(42.24,33.63)(0.04,0.47){1}{\line(0,1){0.47}}
\multiput(42.2,33.16)(0.04,0.47){1}{\line(0,1){0.47}}
\multiput(42.15,32.69)(0.04,0.47){1}{\line(0,1){0.47}}
\multiput(42.11,32.23)(0.05,0.46){1}{\line(0,1){0.46}}
\multiput(42.06,31.77)(0.05,0.46){1}{\line(0,1){0.46}}
\multiput(42,31.32)(0.05,0.45){1}{\line(0,1){0.45}}
\multiput(41.95,30.88)(0.06,0.45){1}{\line(0,1){0.45}}
\multiput(41.89,30.43)(0.06,0.44){1}{\line(0,1){0.44}}
\multiput(41.83,30)(0.06,0.43){1}{\line(0,1){0.43}}
\multiput(41.77,29.57)(0.06,0.43){1}{\line(0,1){0.43}}
\multiput(41.7,29.15)(0.07,0.42){1}{\line(0,1){0.42}}
\multiput(41.63,28.73)(0.07,0.42){1}{\line(0,1){0.42}}
\multiput(41.56,28.33)(0.07,0.41){1}{\line(0,1){0.41}}
\multiput(41.49,27.92)(0.07,0.4){1}{\line(0,1){0.4}}
\multiput(41.41,27.53)(0.08,0.39){1}{\line(0,1){0.39}}
\multiput(41.33,27.14)(0.08,0.39){1}{\line(0,1){0.39}}
\multiput(41.25,26.77)(0.08,0.38){1}{\line(0,1){0.38}}
\multiput(41.17,26.4)(0.08,0.37){1}{\line(0,1){0.37}}
\multiput(41.08,26.04)(0.09,0.36){1}{\line(0,1){0.36}}
\multiput(40.99,25.68)(0.09,0.35){1}{\line(0,1){0.35}}
\multiput(40.9,25.34)(0.09,0.34){1}{\line(0,1){0.34}}
\multiput(40.81,25)(0.09,0.33){1}{\line(0,1){0.33}}
\multiput(40.71,24.68)(0.09,0.33){1}{\line(0,1){0.33}}
\multiput(40.62,24.36)(0.1,0.32){1}{\line(0,1){0.32}}
\multiput(40.52,24.06)(0.1,0.31){1}{\line(0,1){0.31}}
\multiput(40.42,23.76)(0.1,0.3){1}{\line(0,1){0.3}}
\multiput(40.32,23.48)(0.1,0.29){1}{\line(0,1){0.29}}
\multiput(40.21,23.2)(0.1,0.28){1}{\line(0,1){0.28}}
\multiput(40.11,22.93)(0.11,0.27){1}{\line(0,1){0.27}}
\multiput(40,22.68)(0.11,0.25){1}{\line(0,1){0.25}}
\multiput(39.89,22.44)(0.11,0.24){1}{\line(0,1){0.24}}
\multiput(39.78,22.2)(0.11,0.23){1}{\line(0,1){0.23}}
\multiput(39.67,21.98)(0.11,0.22){1}{\line(0,1){0.22}}
\multiput(39.56,21.77)(0.11,0.21){1}{\line(0,1){0.21}}
\multiput(39.44,21.57)(0.11,0.2){1}{\line(0,1){0.2}}
\multiput(39.33,21.38)(0.12,0.19){1}{\line(0,1){0.19}}
\multiput(39.21,21.21)(0.12,0.18){1}{\line(0,1){0.18}}
\multiput(39.09,21.04)(0.12,0.16){1}{\line(0,1){0.16}}
\multiput(38.97,20.89)(0.12,0.15){1}{\line(0,1){0.15}}
\multiput(38.85,20.75)(0.12,0.14){1}{\line(0,1){0.14}}
\multiput(38.73,20.62)(0.12,0.13){1}{\line(0,1){0.13}}
\multiput(38.61,20.5)(0.12,0.12){1}{\line(1,0){0.12}}
\multiput(38.49,20.4)(0.12,0.1){1}{\line(1,0){0.12}}
\multiput(38.37,20.3)(0.12,0.09){1}{\line(1,0){0.12}}
\multiput(38.25,20.22)(0.12,0.08){1}{\line(1,0){0.12}}
\multiput(38.12,20.16)(0.12,0.07){1}{\line(1,0){0.12}}
\multiput(38,20.1)(0.12,0.06){1}{\line(1,0){0.12}}
\multiput(37.87,20.06)(0.12,0.04){1}{\line(1,0){0.12}}
\multiput(37.75,20.02)(0.12,0.03){1}{\line(1,0){0.12}}
\multiput(37.62,20.01)(0.12,0.02){1}{\line(1,0){0.12}}
\multiput(37.5,20)(0.12,0.01){1}{\line(1,0){0.12}}
\multiput(37.38,20.01)(0.12,-0.01){1}{\line(1,0){0.12}}
\multiput(37.25,20.02)(0.12,-0.02){1}{\line(1,0){0.12}}
\multiput(37.13,20.06)(0.12,-0.03){1}{\line(1,0){0.12}}
\multiput(37,20.1)(0.12,-0.04){1}{\line(1,0){0.12}}
\multiput(36.88,20.16)(0.12,-0.06){1}{\line(1,0){0.12}}
\multiput(36.75,20.22)(0.12,-0.07){1}{\line(1,0){0.12}}
\multiput(36.63,20.3)(0.12,-0.08){1}{\line(1,0){0.12}}
\multiput(36.51,20.4)(0.12,-0.09){1}{\line(1,0){0.12}}
\multiput(36.39,20.5)(0.12,-0.1){1}{\line(1,0){0.12}}
\multiput(36.27,20.62)(0.12,-0.12){1}{\line(1,0){0.12}}
\multiput(36.15,20.75)(0.12,-0.13){1}{\line(0,-1){0.13}}
\multiput(36.03,20.89)(0.12,-0.14){1}{\line(0,-1){0.14}}
\multiput(35.91,21.04)(0.12,-0.15){1}{\line(0,-1){0.15}}
\multiput(35.79,21.21)(0.12,-0.16){1}{\line(0,-1){0.16}}
\multiput(35.67,21.38)(0.12,-0.18){1}{\line(0,-1){0.18}}
\multiput(35.56,21.57)(0.12,-0.19){1}{\line(0,-1){0.19}}
\multiput(35.44,21.77)(0.11,-0.2){1}{\line(0,-1){0.2}}
\multiput(35.33,21.98)(0.11,-0.21){1}{\line(0,-1){0.21}}
\multiput(35.22,22.2)(0.11,-0.22){1}{\line(0,-1){0.22}}
\multiput(35.11,22.44)(0.11,-0.23){1}{\line(0,-1){0.23}}
\multiput(35,22.68)(0.11,-0.24){1}{\line(0,-1){0.24}}
\multiput(34.89,22.93)(0.11,-0.25){1}{\line(0,-1){0.25}}
\multiput(34.79,23.2)(0.11,-0.27){1}{\line(0,-1){0.27}}
\multiput(34.68,23.48)(0.1,-0.28){1}{\line(0,-1){0.28}}
\multiput(34.58,23.76)(0.1,-0.29){1}{\line(0,-1){0.29}}
\multiput(34.48,24.06)(0.1,-0.3){1}{\line(0,-1){0.3}}
\multiput(34.38,24.36)(0.1,-0.31){1}{\line(0,-1){0.31}}
\multiput(34.29,24.68)(0.1,-0.32){1}{\line(0,-1){0.32}}
\multiput(34.19,25)(0.09,-0.33){1}{\line(0,-1){0.33}}
\multiput(34.1,25.34)(0.09,-0.33){1}{\line(0,-1){0.33}}
\multiput(34.01,25.68)(0.09,-0.34){1}{\line(0,-1){0.34}}
\multiput(33.92,26.04)(0.09,-0.35){1}{\line(0,-1){0.35}}
\multiput(33.83,26.4)(0.09,-0.36){1}{\line(0,-1){0.36}}
\multiput(33.75,26.77)(0.08,-0.37){1}{\line(0,-1){0.37}}
\multiput(33.67,27.14)(0.08,-0.38){1}{\line(0,-1){0.38}}
\multiput(33.59,27.53)(0.08,-0.39){1}{\line(0,-1){0.39}}
\multiput(33.51,27.92)(0.08,-0.39){1}{\line(0,-1){0.39}}
\multiput(33.44,28.33)(0.07,-0.4){1}{\line(0,-1){0.4}}
\multiput(33.37,28.73)(0.07,-0.41){1}{\line(0,-1){0.41}}
\multiput(33.3,29.15)(0.07,-0.42){1}{\line(0,-1){0.42}}
\multiput(33.23,29.57)(0.07,-0.42){1}{\line(0,-1){0.42}}
\multiput(33.17,30)(0.06,-0.43){1}{\line(0,-1){0.43}}
\multiput(33.11,30.43)(0.06,-0.43){1}{\line(0,-1){0.43}}
\multiput(33.05,30.88)(0.06,-0.44){1}{\line(0,-1){0.44}}
\multiput(33,31.32)(0.06,-0.45){1}{\line(0,-1){0.45}}
\multiput(32.94,31.77)(0.05,-0.45){1}{\line(0,-1){0.45}}
\multiput(32.89,32.23)(0.05,-0.46){1}{\line(0,-1){0.46}}
\multiput(32.85,32.69)(0.05,-0.46){1}{\line(0,-1){0.46}}
\multiput(32.8,33.16)(0.04,-0.47){1}{\line(0,-1){0.47}}
\multiput(32.76,33.63)(0.04,-0.47){1}{\line(0,-1){0.47}}
\multiput(32.72,34.1)(0.04,-0.47){1}{\line(0,-1){0.47}}
\multiput(32.69,34.58)(0.04,-0.48){1}{\line(0,-1){0.48}}
\multiput(32.65,35.06)(0.03,-0.48){1}{\line(0,-1){0.48}}
\multiput(32.63,35.55)(0.03,-0.48){1}{\line(0,-1){0.48}}
\multiput(32.6,36.04)(0.03,-0.49){1}{\line(0,-1){0.49}}
\multiput(32.58,36.53)(0.02,-0.49){1}{\line(0,-1){0.49}}
\multiput(32.56,37.02)(0.02,-0.49){1}{\line(0,-1){0.49}}
\multiput(32.54,37.51)(0.02,-0.49){1}{\line(0,-1){0.49}}
\multiput(32.52,38.01)(0.01,-0.5){1}{\line(0,-1){0.5}}
\multiput(32.51,38.51)(0.01,-0.5){1}{\line(0,-1){0.5}}
\multiput(32.51,39)(0.01,-0.5){1}{\line(0,-1){0.5}}
\multiput(32.5,39.5)(0,-0.5){1}{\line(0,-1){0.5}}
\multiput(32.5,40)(0,-0.5){1}{\line(0,-1){0.5}}
\put(32.5,40){\vector(-0,1){0.12}}

\put(7.5,62.5){\makebox(0,0)[cc]{Send(m)}}

\put(0,27.5){\makebox(0,0)[cc]{SendPacket(m,ab)}}

\put(52.5,30){\makebox(0,0)[cc]{}}

\put(132.5,55){\makebox(0,0)[cc]{}}

\put(117.5,15){\makebox(0,0)[cc]{Physical Channel}}

\put(117.5,50){\makebox(0,0)[cc]{Stabilizing Data-Link}}

\put(117.5,75){\makebox(0,0)[cc]{Application Layer}}

\linethickness{0.3mm}
\put(97.5,80){\line(1,0){40}}
\put(97.5,70){\line(0,1){10}}
\put(137.5,70){\line(0,1){10}}
\put(97.5,70){\line(1,0){40}}
\linethickness{0.3mm}
\put(97.5,55){\line(1,0){40}}
\put(97.5,40){\line(0,1){15}}
\put(137.5,40){\line(0,1){15}}
\put(97.5,40){\line(1,0){40}}
\linethickness{0.3mm}
\put(97.5,20){\line(1,0){40}}
\put(97.5,10){\line(0,1){10}}
\put(137.5,10){\line(0,1){10}}
\put(97.5,10){\line(1,0){40}}
\linethickness{0.3mm}
\put(127.5,55){\line(0,1){15}}
\put(127.5,70){\vector(0,1){0.12}}
\linethickness{0.3mm}
\put(107.5,20){\line(0,1){20}}
\put(107.5,20){\vector(0,-1){0.12}}
\put(147.5,62.5){\makebox(0,0)[cc]{DeliverMessage(m)}}

\linethickness{0.3mm}
\put(92.5,85){\line(1,0){50}}
\put(92.5,5){\line(0,1){80}}
\put(142.5,5){\line(0,1){80}}
\put(92.5,5){\line(1,0){50}}
\linethickness{0.3mm}
\put(2.5,85){\line(1,0){50}}
\put(2.5,5){\line(0,1){80}}
\put(52.5,5){\line(0,1){80}}
\put(2.5,5){\line(1,0){50}}
\put(27.5,0){\makebox(0,0)[cc]{$p_i$}}

\put(117.5,0){\makebox(0,0)[cc]{$p_j$}}

\linethickness{0.3mm}
\put(47.5,17.5){\line(1,0){50}}
\put(97.5,17.5){\vector(1,0){0.12}}
\put(27.5,45){\makebox(0,0)[cc]{(Send part)}}

\put(117.5,45){\makebox(0,0)[cc]{(Receive part)}}

\put(67.5,35){\makebox(0,0)[cc]{ReceivePacket(ack,(m,ab))}}

\put(85,27.5){\makebox(0,0)[cc]{SendPacket(ack,(m,ab))}}

\put(147.5,35){\makebox(0,0)[cc]{ReceivePacket(m,ab)}}

\linethickness{0.3mm}
\put(47.5,12.5){\line(1,0){50}}
\put(47.5,12.5){\vector(-1,0){0.12}}
\put(72.5,20){\makebox(0,0)[cc]{Messages}}

\put(72.5,10){\makebox(0,0)[cc]{Acknowledgements}}

\linethickness{0.3mm}
\put(127.5,20){\line(0,1){20}}
\put(127.5,40){\vector(0,1){0.12}}
\linethickness{0.3mm}
\put(37.5,55){\line(0,1){15}}
\put(37.5,70){\vector(0,1){0.12}}
\put(52.5,62.5){\makebox(0,0)[cc]{DeliverAck(m)}}

\linethickness{0.3mm}
\put(107.5,55){\line(0,1){15}}
\put(107.5,55){\vector(0,-1){0.12}}
\put(97.5,62.5){\makebox(0,0)[cc]{Receive()}}

\end{picture}